\documentclass[fleqn]{article} 
\usepackage{amssymb}
\usepackage{amsmath,amsfonts,amsthm}

\let\svthefootnote\thefootnote
\newcommand\blankfootnote[1]{%
  \let\thefootnote\relax\footnotetext{#1}%
  \let\thefootnote\svthefootnote%
}
\let\svfootnote\footnote
\renewcommand\footnote[2][?]{%
  \if\relax#1\relax%
    \blankfootnote{#2}%
  \else%
    \if?#1\svfootnote{#2}\else\svfootnote[#1]{#2}\fi%
  \fi
}

\newtheorem{theorem}{Theorem}[section]
\newtheorem{lemma}[subsection]{Lemma}

\newtheorem{remark}[theorem]{Remark}
\newtheorem{definition}[theorem]{Definition}
\newtheorem{example}[theorem]{Example}

\usepackage{fullpage}

\begin{document}

\numberwithin{equation}{section}

\title{Static Einstein-Maxwell space-time invariant by translation}

\author{$^{\ast}$Benedito Leandro$^1$, Ana Paula de Melo$^2$, Ilton Menezes$^3$ and Romildo Pina$^4$}

\footnote[]{Universidade Federal de Goi\'as, IME, CEP 74690-900, Goi\^ania, GO, Brazil.}
\footnote[]{Email address: \textsf{$^{\ast}$bleandroneto@ufg.br$^1$, anapmelocosta@gmail.com$^2$, iltomenezesufg@gmail.com$^3$, romildo@ufg.br$^4$.}}
	 \footnote[]{Ana Paula de Melo was partially supported by PROPG-CAPES [Finance Code 001].}
	 \footnote[]{Ilton Menezes was partially supported by PROPG-CAPES [Finance Code 001].}
	 \footnote[]{Romildo Pina was partially supported by CNPq/Brazil [Grant number: 305410/2018-0].}

\date{}

\maketitle{}

\begin{abstract}
In this paper we study the static Einstein-Maxwell space when it is conformal to an $n$-dimensional pseudo-Euclidean space, which is invariant under the action of an $(n-1)$-dimensional translation group. We also provide a complete classification of such space. 
\end{abstract}

\vspace{0.2cm} \noindent \emph{2020 Mathematics Subject
Classification} : 53C20; 83C22; 53C21; 30F45. 

\vspace{0.4cm}\noindent \emph{Keywords}: Electrostatic system, conformal metric, Einstein-Maxwell equations.

\ 

\section{Introduction and Main Results}

 \paragraph{} We can classify black holes by three aspects: mass, electric charge and angular momentum (at least in the static case). From a physical point of view, the most interesting and well-known solutions for Einstein's Equations, which represent a physical model for a black hole, are: Schwarzschild, Reissner-Nordstr\"om and Kerr-Newman space-times. The first one was also the first nontrivial solution exact solution for the Einstein's equation and models a black hole with mass that do not rotate and has no electric charge. The Kerr-Newman solution is a model for a black hole with mass, electric charge and angular momentum. Since it is expected that a black hole rotates, given the nature and final state of a collapsing star, this is a more natural solution to consider given its properties. Nevertheless,  the Reissner-Nordstr\"om solution represents a class of black holes which carries mass and electric charge, and
thus can be very useful in the understanding of the theory of general relativity.
 
The latter solution corresponds to a model for a static black hole with electric charged $q$ and mass $m$, which is spherically symmetric and conformally flat. The spatial factor for this space-time is defined by the Riemannian manifold $M^{n}=\mathbb{S}^{n-1}\times (r^{+},\,+\infty)$ with metric tensor
$$g=\frac{dr^{2}}{1-mr^{2-n}+q^{2}r^{2(2-n)}}+r^{2}g_{\mathbb{S}^{n-1}},$$
where $m>2q>0$ are constants (subextremal), and $r^{+}$ is defined as the larger of the two solutions of the equation $1-mr^{2-n}+q^{2}r^{2(2-n)}=0.$
The Reissner-Nordstr\"om manifold is one of the most relevant solutions for the electrostatic system (or Einstein-Maxwell equations), which we will define below (cf. \cite{cederbaum2016,chrusciel,sophia,yaza2015} and the references therein). The static Einstein-Maxwell space-time allows the existence of two distinct black holes in equilibrium, due to the electric charge. This is an important distinction from static vacuum Einstein space-times. The electrostatic system generalizes the static vacuum Einstein equations. We are considering there are no magnetic fields in the system.

\begin{definition}\label{def1} Let $(M^n, g)$, $n \geq 3$, be a semi-Riemannian manifold and let $N: M^n\rightarrow \mathbb R_{>0}$, $\psi: M^n\rightarrow \mathbb R$ be smooth functions such that
\begin{equation}\label{001}
\Delta N =\widehat{C}^{2}\frac{|\nabla\psi|^2}{N},
\end{equation}
\begin{equation}\label{002}
\operatorname{div} \left(\frac{\nabla\psi}{N}\right)=0,
\end{equation}
and
\begin{equation}\label{003}
N\operatorname{Ric} =\nabla^2 N-2\frac{\nabla\psi\otimes \nabla\psi}{N}+\frac{2}{(n-1)N}|\nabla\psi|^2g,
\end{equation}
where $\widehat{C}^{2}=2\frac{(n-2)}{(n-1)}$. Then $(M^n,\, g,\, N,\, \psi)$ is called an \emph{electrostatic system} (or static Einstein-Maxwell equations). 
Here $\nabla$, $\Delta$ and $\nabla^2$ stand for the covariant derivative, the Laplacian and the Hessian operator with respect to $g$. Moreover, $\operatorname{Ric}$ and $\operatorname{div}$ are the Ricci tensor and the divergence for $g$. The smooth functions $N$, $\psi$ and $M^{n}$ are called \textit{lapse  function}, \textit{electric potential} and spatial factor (in the Riemannian case) for the static Einstein-Maxwell space-time, respectively. 
\end{definition}

It is important to notice that if $\psi$ is constant everywhere, then we get the {\it static vacuum Einstein equations}, for which the standard Schwarzschild space is the most important solution  (see \cite{santos} for instance). Contracting equation \eqref{003} and combining it with \eqref{001},  we get
\begin{equation}\label{1.4}
N^2R=2|\nabla \psi|^2,
\end{equation}
where $R$ denotes the scalar curvature of $(M^n,\, g)$ (for more details see \cite{cederbaum2016,chrusciel,sophia}).

The most complete result there exists in the literature so far concerning the classification of electrostatic spaces are due Chru\'sciel and Tod \cite{chrusciel2007}, in the three dimensional case and under some asymptotic conditions.

 Our goal is to provide a full classification for the solutions of static Einstein-Maxwell equations when $M^n$ is conformal to a pseudo-Euclidean space and invariant under the action of an $(n-1)$-dimensional translation group. To do that, let $(\Omega, \, g)$ be the standard pseudo-Euclidean space with metric $g$ and Cartesian coordinates $(x_1, ..., x_n)$ in which $g_{ij}=\delta_{ij}\varepsilon_i$, $1\leq i,j\leq n$. Here, $\delta_{ij}$ is the Kronecker delta and $\varepsilon_{i}=\pm1$. Let $\xi=\sum_{i=1}^{n}\alpha_{i}x_{i}$, $\alpha_{i}\in\mathbb{R}$ be a basic invariant for an $(n-1)$-dimensional translation group. We obtain necessary and sufficient conditions for smooth functions $\varphi$, $N$ and $\psi$ defined in $\Omega\subseteq\mathbb{R}^n$, an open subset, to be invariant by the action of an $(n-1)$-dimensional translation group so that $ (\Omega, \, \bar{g},\,N,\,\psi)$ be a solution for the electrostatic system with metric $\bar{g}=\frac{g}{\varphi^2}$. Moreover, we show that these necessary and sufficient conditions depend on the
direction of $\alpha=\displaystyle\sum_{i=1}^{n}\alpha_i \partial_{x_i}$ (i.e., being lightlike or not).

This method was successfully used before to provide a complete classification for Solitons, Static metrics and quasi-Eisntein manifolds. In \cite{barbosa}, Barbosa, Pina and Tenenblat showed a complete classification for the gradient Ricci solitons. Moreover, they proved there are infinitely many gradient Ricci solitons conformal to the standard pseudo-Euclidean metric. For the gradient Yamabe solitons, a similar result was proved by Leandro and Tenenblat \cite{Be1}, where a geodesically complete example of Yamabe soliton was provided. Leandro and dos Santos \cite{Be2} proved the most general invariant that reduces the gradient Ricci solitons PDE into a system of ODE. Furthermore, they built new examples in the two dimensional case. The idea of finding the most general invariant to reduce the Ricci solitons was also used by dos Santos and Leandro in \cite{santos}, where the authors were able to recover the traditional Tangherlini-Schwarzschild solution. Recently, Pina and dos Santos \cite{pina} proved all group-invariants for the Einstein equation and the Ricci equation. Ribeiro Jr and Tenenblat \cite{ribeiro} provided a complete classification for quasi-Einstein metrics conformal to a Euclidean space and invariant by translations.

It is important to highlight that we are considering the static Einstein-Maxwell space conformal to an $n$-dimensional pseudo-Euclidean space, which is invariant under the action of an $(n-1)$-dimensional translation group. Therefore, our solutions, when we consider the Riemannian case for the spatial factor, can not be geodesically complete, since the conformal function is not defined for all values of its parameter. Furthermore, the solutions presented in this work must be different from the Reissner-Nordstr\"om manifold since it is a conformally flat and rotationally symmetric solution for the electrostatic system (cf. \cite{chrusciel}). Still, from a geometric point of view it is important to obtain as much examples as possible, for a better understanding of the geometry and topology of solutions for the electrostatic system.

The classification of electrostatic solutions for the Einstein's equation was already provided in the literature for the Riemannian case, considering some suitable initial boundary conditions (cf.\cite{cederbaum2016,chrusciel,sophia} and the references therein). We can follow, for instance, the steps of the proof given by \cite{bunting} to conclude that an asymptotically flat static Einstein-Maxwell space must be the Reissner-Nordstr\"om manifold (cf. \cite{sophia} for a full description of this method). It is important to say that the Reissner-Nordstr\"om manifold is spherically symmetric, and so our solutions must be completely different from it. Moreover, we are considering the manifold with a semi-Riemannian geometry. Thus, our solutions will be more general, in the sense that even the topology of our solutions can be different.

In what follows, we will use the following convention for the  derivatives of a function $F=F(\xi)$, where $\xi: \Omega\subseteq\mathbb{R}^n\rightarrow\mathbb{R}$:
\ 
\
\
\begin{center}
    $\frac{d F}{d\xi}=F'$\quad\mbox{and}\quad
    $\frac{\partial F}{\partial x_i\partial x_j}=F_{,ij}$.
\end{center}

 \begin{theorem}\label{teo1}
Let $\left(\mathbb{R}^{n},g\right)$, $\, n\geq 3$, be a pseudo-Euclidean space with Cartesian coordinates $x=\left(x_{1},...\, ,x_{n}\right)$ and metric components $g_{ij}=\delta_{ij}\varepsilon_{i}$, $1\leq i,j\leq n$, where $\delta_{ij}$ is the Kronecker delta and $\varepsilon_{i}=\pm1$. Then there exists a metric $\bar{g}=g/ \varphi^{2}$ such that $\left(\Omega,\,\bar{g},\,N,\,\psi\right)$, $\Omega\subseteq\mathbb{R}^{n}$, is a solution for the electrostatic system (Definition \ref{def1}) if, and only if, smooth functions $\varphi$, $\psi$ and $N$ satisfy 

\begin{equation}\label{004}
(n-2)N\varphi_{,ij}-\varphi N_{,ij}-\varphi_{,i}N_{,j}-\varphi_{,j}N_{,i}+2\frac{\varphi}{N}\psi_{,i}\psi_{,j}=0,\hspace{0,8cm} \mathrm{for} \hspace{0,2cm} i\neq j;
\end{equation}
and for each $i$
\begin{eqnarray}\label{005}
&&\varphi\left[(n-2)N\varphi_{,ii}-\varphi N_{,ii}-2\varphi_{,i}N_{,i}+2\frac{\varphi}{N}\left(\psi_{,i}\right)^2 \right]\nonumber\\
&&+\varepsilon_{i}\sum_{k=1}^{n}\varepsilon_{k}\left[\varphi\varphi_{,kk}N+\varphi\varphi_{,k}N_{,k}-(n-1)N\left(\varphi_{,k}\right)^2-\frac{2}{(n-1)N}\varphi^2\left(\psi_{,k}\right)^2\right]=0;
\end{eqnarray}

\begin{equation}\label{19}
\sum_{k=1}^{n}\varepsilon_k\{N\varphi \psi_{,kk}-(n-2)N\varphi_{,k}\psi_{,k}-\varphi\psi_{,k}N_{,k}\}=0;
\end{equation}

\begin{equation}\label{018}
\sum_{k=1}^{n}\varepsilon_k\{\varphi NN_{,kk} -(n-2)N \varphi_{,k}N_{,k}-\widehat C^2\varphi \left(\psi_{,k}\right)^2\}=0.
\end{equation}
\end{theorem}

We aim to find solutions for the PDE's \eqref{004}, \eqref{005}, \eqref{19} and \eqref{018} of the form $\varphi(\xi), \psi(\xi)$ and $N(\xi)$, where $\xi=\sum_{i=1}^{n}\alpha_{i}x_{i}$, $\alpha_{i}\in\mathbb{R}$, is an invariant for an $ (n-1)$-dimensional translation group with $\sum_{i=1}^n\varepsilon_i\alpha_i^2=\varepsilon_{i_0}$, where $\varepsilon_{i_0}\in \{-1,\,1,\,0\}$ (timelike, spacelike or lightlike).

 The next result reduces the system of partial differential equations from the previous theorem to an equivalent system of ordinary differential equations.

\begin{theorem}\label{teo2}
Let $\left(\mathbb{R}^{n},g\right)$, $n\geq 3$, be a pseudo-Euclidean space with Cartesian coordinates $x=\left(x_{1},...,x_{n}\right)$ and metric components $g_{ij}=\delta_{ij}\varepsilon_{i}$, $1\leq i,j \leq n$, where $\varepsilon_{i}=\pm 1$. Consider smooth functions $\varphi(\xi), \psi(\xi)$ and $N(\xi)$, where $\xi=\sum_{i=1}^{n}\alpha_{i}x_{i}$, $\alpha_{i}\in\mathbb{R}$ and $\sum_{i=1}^n\varepsilon_i\alpha_i^2=\varepsilon_{i_0}$. Then there exists a metric $\bar{g}=g/\varphi^2$ such that $\left(\Omega,\,\bar{g},\,N,\,\psi\right)$, $\Omega\subseteq\mathbb{R}^{n}$,
is a solution for the electrostatic system (Definition \ref{def1}) if, and only if, functions $\varphi$, $\psi$ and $N$ satisfy
\vspace{12pt}
\begin{align}\label{3}
(n-2)\varphi''N-\varphi N''-2\varphi'N'+\frac{2\varphi}{N}\left(\psi'\right)^2=0;
\end{align}

\begin{align}\label{4}
\varepsilon_{i_0}\left\{\varphi\varphi''N-(n-1)N\left(\varphi'\right)^2+\varphi\varphi'N'-\frac{2\varphi^2\left(\psi'\right)^2}{(n-1)N}\right\}=0;
\end{align}

\begin{align}\label{20}
\varepsilon_{i_0}\left\{\varphi N\psi''-(n-2)\varphi'N\psi'-\varphi N'\psi'\right\}=0;
\end{align}

\begin{align}\label{21}
\varepsilon_{i_0}\left\{\varphi NN''-(n-2)\varphi'NN'-\widehat{C}^2\varphi\left(\psi'\right)^2\right\}=0.
\end{align}

Moreover, if  $\sum_{i=1}^n\varepsilon_i\alpha_i^2=0$, then
\begin{align*}
(n-2)\varphi''N-\varphi N''-2\varphi'N'+\frac{2\varphi}{N}\left(\psi'\right)^2=0.
\end{align*}
\end{theorem}

The following theorem will show that there are infinitely many solutions for electrostatic systems, which are invariant under the action of an $(n-1)$-dimensional translation group acting on $\Omega\subseteq\mathbb{R}^n$, when $\alpha=\displaystyle\sum_{i=1}^{n}\alpha_{i}\partial_{x_i}$ is a null
vector (i.e., $\sum_{i=1}^n\varepsilon_i\alpha_i^2=0)$. 

\begin{theorem}\label{theorem 6} Let $\varphi(\xi)$ and $N(\xi)$ be any nonvanishing differentiable functions, where  $\xi=\sum_{i=1}^{n}\alpha_{i}x_{i}$ and $\sum_{i=1}^n\varepsilon_i\alpha_i^2=0$. Then the function $\psi(\xi)$ given by:
	\begin{equation}\label{ps}
	 \psi(\xi)= \pm \int \sqrt{N \left( \frac{N''}{2} + \frac{\varphi' \, N'}{\varphi} - (n-2) \frac{N \, \varphi''}{2\varphi} \right)} \, \, d\xi+ c,
	\end{equation}
	where $c \in \mathbb{R}$, satisfies \eqref{3}. Moreover,  $(\Omega,\,\bar{g},\,N,\,\psi)$, in which $\Omega\subseteq\mathbb{R}^{n}$ and $\bar{g}=\varepsilon_i\delta_{ij} /\varphi^{2}$, is a solution for the electrostatic system. All such spaces given by these solutions have zero scalar curvature.
\end{theorem}

Next, we provide a necessary condition satisfied by the conformal factor $\varphi$ in order that $\left(\Omega,\,\bar{g},\,N,\,\psi\right)$, $\Omega\subseteq\mathbb{R}^{n}$, be a solution for the electrostatic system with $\bar{g}=g/\varphi^2$.

\begin{theorem}\label{theorem 7}
    Let $\left(\mathbb{R}^{n},g\right)$, $n\geq 3$, be a pseudo-Euclidean space with Cartesian coordinates $x=\left(x_{1},...,x_{n}\right)$ and metric components $g_{ij}=\delta_{ij}\varepsilon_{i}$, $1\leq i,j \leq n$, where $\varepsilon_{i}=\pm 1$. Consider smooth functions $\varphi(\xi), \psi(\xi)$ and $N(\xi)$, where $\xi=\sum_{i=1}^{n}\alpha_{i}x_{i}$, $\alpha_{i}\in\mathbb{R}$ and $\sum_{i=1}^n\varepsilon_i\alpha_i^2=\varepsilon_{i_0}$, with $\varepsilon_{i_0} \in \{-1,1\}$. Then, if $\left(\Omega,\,\bar{g},\,N,\,\psi\right)$, $\Omega\subseteq\mathbb{R}^{n}$, is a solution for the electrostatic system
\begin{equation}\label{edophi}
     \varphi^{2}\varphi'''  - 3(n-1) \varphi\varphi' \varphi''+ n(n-1)(\varphi')^{3} = 0.
\end{equation}
Moreover,
    \begin{eqnarray}\label{Npsidependephi}
    N= k \frac{\varphi'}{\varphi}\quad\mbox{and}\quad
     \psi= k_1 \frac{\varphi^{n-2}}{n-2} + k_2.
    \end{eqnarray} 
Here, $k, k_1 \in \mathbb{R}\backslash\{0\}$, $k_2\in\mathbb{R}$ and $\bar{g}=g/\varphi^2$.
\end{theorem}

In what follows, we classify the static Einstein-Maxwell space when it is conformal to a pseudo-Euclidean space and invariant under the action of a translation group. 

\begin{theorem}\label{cor1}
		Let $\left(\mathbb{R}^{n},g\right)$, $n\geq3$, be a pseudo-Euclidean space with  Cartesian coordinates $x=\left(x_{1},...,x_{n}\right)$ and $g_{ij}=\delta_{ij}\varepsilon_{i}$, $1\leq i,j \leq n$, where $\varepsilon_{i}=\pm 1$. Consider smooth functions $\varphi(\xi), \psi(\xi)$ and $N(\xi)$, where $\xi=\sum_{i=1}^{n}\alpha_{i}x_{i}$, $\alpha_{i}\in\mathbb{R}$ and $\sum_{i=1}^{n}\varepsilon_{i}\alpha_i^2=\varepsilon_{i_0}$, with $\varepsilon_{i_0}=\pm1$. Then $(\Omega, \, \bar{g},\,N,\,\psi)$, $\Omega\subseteq\mathbb{R}^n$, is a solution for the electrostatic system with  $\bar{g}=g/\varphi^2$ if, and only if, functions $\varphi$, $\psi$ and $N$ satisfy
\begin{equation}\label{varphi}
  \varphi(\xi)=c_3[1- c_{1}^{2}[c_2+(n-2)\xi]^2]^{-1/(n-2)},
\end{equation}

\begin{equation}\label{psi}
   \psi(\xi)=\frac{k_1 c_3^{n-2}}{n-2}[1- c_{1}^{2}[c_2+(n-2)\xi]^2]^{-1} + k_2
\end{equation}

and

\begin{equation}\label{N}
     N\left(\xi\right)= \frac{2kc_{1}^{2}[c_2 + (n-2)\xi]}{1 - c_{1}^{2}[c_2+(n-2)\xi]^2},
\end{equation}
where $c_1^2= \frac{ k_1^{2}c_3^{2(n-2)}}{2(n-1)(n-2)k^{2}}$. Here $ k\neq0,\,k_1\neq0,\,k_2,\, c_1\neq0,\,c_2,\,c_3\neq0$ are constants. Moreover, these solutions are  defined in $\Omega=\{1 - c_{1}^{2}[c_2+(n-2)\xi]^2>0\}\bigcap\{2kc_{1}^{2}[c_2 + (n-2)\xi]>0\}$. Considering the Riemannian case (i.e., $\varepsilon_i=1$ for $1\leq i\leq n$), we can conclude that this family of solutions can not be geodesically complete.
\end{theorem}

\begin{remark}
Here, we provide the curvature for Theorem \ref{cor1}. In the Appendix the reader will find the sectional curvature formulas to our settings. For Theorem \ref{cor1}, the sectional curvature is
\begin{eqnarray*}
K_{i\,j}=\frac{ 2c_{1}^{2}c_{3}^{2}\left\{c_{1}^{2}(c_{2}+(n-2)\xi)^{2}[n(\varepsilon_{i}\alpha_{i}^{2}+\varepsilon_{j}\alpha_{j}^{2})-2\varepsilon_{i_0}]+(n-2)(\varepsilon_{i}\alpha_{i}^{2}+\varepsilon_{j}\alpha_{j}^{2})\right\}}{[1- c_{1}^{2}[c_2+(n-2)\xi]^2]^{2(n-1)/(n-2)}}.
\end{eqnarray*}

Moreover, a straightforward computation (see \eqref{1.4} for instance) gives us the scalar curvature
\begin{eqnarray*}
R=\frac{4(n-1)(n-2)\varepsilon_{i_0}c_{1}^{2}c_{3}^{2}}{[1- c_{1}^{2}[c_2+(n-2)\xi]^2]^{2(n-1)/(n-2)}}.
\end{eqnarray*}
Thus, we can see that if $\xi\rightarrow \dfrac{-c_{2}c_{1}\pm1}{(n-2)c_1}$ the scalar curvature $R$ goes to infinity. Moreover, if $\xi\rightarrow\pm\infty$ then $R\rightarrow 0$.
\end{remark}

 It is important to remember that higher dimensional solutions for Einstein's equation called physicists' attention with the development of string theory. Therefore, solutions for the Einstein's equation with non conventional signatures can represent a step forward in the analysis of general relativity.

\section{Examples for the Einstein-Maxwell equations}

\begin{example}
This example illustrates Theorem \ref{theorem 6}. Let $\varphi(\xi)=\frac{\xi^2}{2}+1$ and $N(\xi)=k\xi^{\frac{n-2}{2}}$, where $k\in\mathbb{R}_{+}^{*}$, with $\xi=\sum_{i=1}^{n}\alpha_ix_i$ and $\sum_{i=1}^{n}\varepsilon_i\alpha_{i}^{2}=0$. Solving \eqref{ps} we obtain 
\begin{equation*}
    \psi(\xi)=\pm k\sqrt{\frac{(n-4)}{2(n-2)}}\xi^{\frac{n-2}{2}}+c,
\end{equation*}
where $c$ is constant. It follows from Theorem \ref{theorem 6} that $(\mathbb{R}^{n},\,g/\varphi^{2},\,N,\,\psi)$ is a solution for the electrostatic system. The sectional curvature (cf. Appendix) is given by 
\begin{equation*}
    K_{ij}=\left(\frac{\xi^2}{2}+1\right)\left(\varepsilon_j\alpha_j^2+\varepsilon_i\alpha_i^2\right).
\end{equation*}

For $n=4$ we get a solution for static vacuum equations, i.e., $\psi=0$. Therefore, in that case we recover a static vacuum Einstein space (see \cite{santos}).
\end{example}

\begin{example}\label{completo} From the hypothesis of Theorem \ref{theorem 6}, consider $\varphi(\xi)=ke^{\xi}$ and $N(\xi)=e^{\frac{n-2}{2}\xi}$ with $k \in \mathbb{R}^{*}_{+}$. Solving \eqref{ps} we obtain 
\begin{equation*}
    \psi(\xi)=\pm\frac{1}{\sqrt{2}}e^{\frac{n-2}{2}\xi}+c,
\end{equation*}
where $c$ is constant. Let us show that $M=\left(\mathbb{R}^n,\bar{g}\right)\times_{N}\mathbb{R}$ with metric tensor $g_{M}=\bar{g}-N^2dt^2$ is geodesically complete, for a specific choice of null vector and signature of the metric. 

It is well known (see \cite{Be} and \cite{O'neil}) that if the curve is a $\gamma(s)=\left(\sigma(s), v(s)\right)$, it is geodesic in $M=B\times_{N}F$ if, and only if,
\begin{itemize} 
    \item [(i)] $\sigma''(s)=g_{F}\left(v'(s),v'(s)\right)N(\sigma(s))\left(\nabla_{g_{B}}N\right)$ in $B$,
    \item [(ii)] $v''(s)=-\frac{2}{N(\sigma(s))}\frac{d\left(N\circ\sigma\right)}{ds}v'(s)$ in $F$.
\end{itemize}
 
Note that 
\begin{equation*}
    \nabla_{\bar{g}}N=\frac{(n-2)k^2}{2}e^{\frac{n+2}{2}\xi}\sum_{i=1}^n\varepsilon_i\alpha_i\frac{\partial}{\partial x_i}.
\end{equation*}
By $(i)$, we get
\begin{equation*}
    \sigma''(s)=\frac{(n-2)k^2}{2}\left(v'(s)\right)^2e^{n\xi}\sum_{i=1}^n\varepsilon_i\alpha_i\frac{\partial}{\partial x_i}.
\end{equation*}
From the structure of the warped product we have assumed, we can infer that $\sigma(s)=\left(x_1(s),...,x_n(s)\right)$ and $v(s)=x_{n+1}(s)$. Thus,  
\begin{equation*}
x_{i}''(s)=\varepsilon_i\alpha_i\frac{(n-2)k^2}{2}\left(v'(s)\right)^2e^{n\xi}\hspace{0.5cm}\forall \hspace{0.5cm}1\leq i\leq n.
\end{equation*}

Now, from $(ii)$ we have 
\begin{equation*}
    v''(s)=-(n-2)\sum_{i=1}^n\alpha_ix_{i}'(s)v'.
\end{equation*}
Making $\varepsilon_1=-1$, $\varepsilon_i=1$ for $i\neq1$ and $\alpha_1=\alpha_2=1$, $\alpha_i=0$ with $3\leq i\leq n$. Thus $\sum_{i=1}^{n}\varepsilon_i\alpha_{i}^2=0$ and $\xi=x_1+x_2$. Consequently, 
\begin{equation}\label{sis}
\left\{\begin{array}{lcl}
x_{1}''(s)=-\frac{(n-2)k^2}{2}\left(v'(s)\right)^2e^{n\xi},\\

x_{2}''(s)=\frac{(n-2)k^2}{2}\left(v'(s)\right)^2e^{n\xi}\\ 

x_{l}''(s)=0, \hspace{0.5cm} for\hspace{0.5cm} 3\leq l\leq n,\\

v''(s)=-(n-2)\left(x_1'(s)+x_2'(s)\right)v'
\end{array}\right.
\end{equation}
From the first and second equations of the above system, we have 
\begin{equation}\label{var}
    x_1'(s)+x_2'(s)=k_1 \hspace{0.2cm} \Longrightarrow \hspace{0.2cm}x_1(s)+x_2(s)=k_1s+k_2,
\end{equation}
where $k_1, k_2\in \mathbb{R}$. Thus, from the fourth equation of \eqref{sis} and \eqref{var} we obtain

\begin{equation*}
    \frac{v''(s)}{v'(s)}=-(n-2)k_1 \Longrightarrow v'(s)=\bar{k}_3e^{-(n-2)k_1s},
\end{equation*}
where $\bar{k}_3$ is a non null constant. Therefore, 
\begin{equation}\label{v}
\left\{\begin{array}{lcl}
v(s)=-\frac{\bar{k}_3}{(n-2)k_1}e^{-(n-2)k_1s}+k_4,\hspace{0.2cm}if\hspace{0.2cm} k_1\neq0\hspace{0.2cm} and \hspace{0.2cm} k_4\in\mathbb{R}, \\
v(s)=k_3s+k_5, \hspace{0.2cm}if\hspace{0.2cm}k_1=0 \hspace{0.2cm} and \hspace{0.2cm} k_3,\,k_5\in \mathbb{R},\\ 
x_l(s)=k_6s+k_7,\hspace{0.2cm}k_6, k_7\in \mathbb{R}, \hspace{0.2cm} 3\leq l\leq n.
\end{array}\right.
\end{equation}

Moreover, from the first equation of \eqref{sis} and \eqref{var}, we get 
\begin{equation}\label{x}
x_1''(s)=-\frac{(n-2)k^2}{2}\left(v'(s)\right)^2e^{n(k_1s+k_2)}.
\end{equation}
Thus, for $k_1\neq 0$ we get 
\begin{equation*}
x_1(s)=
\left\{\begin{array}{lcl}
-\frac{(n-2)k^2\bar{k}_3^2}{2(4-n)^2k_1^2}e^{nk_2}e^{(4-n)k_1s}+k_8s+k_9,\hspace{0.2cm} for \hspace{0.2cm} n\neq 4\hspace{0.2cm} and \hspace{0.2cm}k_8, k_9\in\mathbb{R}, \\
-\frac{k^2\bar{k}_3^2}{2}e^{4k_2}s^2+k_{10}s+k_{11}, for \hspace{0.2cm} n=4\hspace{0.2cm} and \hspace{0.2cm}k_{10}, k_{11}\in\mathbb{R}.
\end{array}\right.
\end{equation*}

On the other hand, if $k_1=0$, from the second equation of \eqref{v} and \eqref{x}, we have 
\begin{equation*}
x_1(s)=-\frac{(n-2)k^2k_3^2}{4}e^{nk_2}s^2+k_{12}s+k_{13},
\end{equation*}
where $k_{12}$ and $k_{13} \in \mathbb{R}$. Furthermore, if $k_3 = 0$, we get $x_1(s)$ is linear.

However, since $x_1(s)+x_2(s)=k_1s+k_2$, we can see that, in any case, all the geodesics are
defined for the entire real line, which means that the warped product manifold $M=\left(\mathbb{R}^n,\bar{g}\right)\times_{N}\mathbb{R}$ with metric tensor $g_{M}=\bar{g}-N^2dt^2$, where $N(\xi)=e^{\frac{n-2}{2}\xi}$ and $\bar{g}=k^{-2}e^{-2\xi}\varepsilon_i\delta_{ij}$, is geodesically complete.  Moreover, the sectional curvature (cf. Appendix) is given by 
\begin{eqnarray*}
K_{ij}=k^2e^{2\xi}\left(\varepsilon_j\alpha_j^2+\varepsilon_i\alpha_i^2\right).\end{eqnarray*}
\end{example}

\section{Proof of the main results}\label{proofs}

Let us start this section with the following Lemma, which will be fundamental in the proof of the main result (Theorem \ref{cor1}). This result deals with a necessary condition which must be satisfied by functions $\varphi$, $N$ and $\psi$ for $(\Omega,\,\bar{g},\,N,\,\psi)$, $\Omega\subseteq\mathbb{R}^{n}$ is an open subset, where $\bar{g}=\frac{\varepsilon_i \delta_{ij}}{\varphi^{2}}$, be a solution for the electrostatic system (Definition \ref{def1}).

\begin{lemma}\label{lemma 1}
	Let $(\mathbb{R}^n,g)$, $n\geq 3$ be a pseudo-Euclidean space with Cartesian coordinates $(x_1,..., x_n)$ and metric components $g_{ij}=\delta_{ij}\varepsilon_i$, $1\leq i,j \leq n$, where $\varepsilon_i=\pm 1$. Consider smooth functions $\varphi(\xi), \psi(\xi)$ and $N(\xi)$. If there exists a metric $\bar{g}=g/\varphi^{2}$ such that $\left(\Omega,\,\bar{g},\,N,\,\psi\right)$, $\Omega\subseteq\mathbb{R}^{n}$, is a solution for the electrostatic system, then the functions must satisfy
	\begin{equation}\label{23}
2(n-1)N \varphi''-n(n-1)\frac{N}{\varphi} (\varphi')^2 - 2\frac{\varphi}{N} (\psi')^2 =0.
	\end{equation}
\end{lemma}
\begin{proof}
From Equation \eqref{1.4} we have $ N^2R_{\bar{g}} = 2 |\nabla_{\bar{g}} \psi|^2$, where $R_{\bar{g}}$ and $\nabla_{\bar{g}}$ denote, respectively, the scalar curvature and covariant derivative for $\bar{g}$.
It is well known that the scalar curvature (cf. \cite{khunel1988}) is given by
\begin{eqnarray*}
 R _ {\bar{g}} = (n-1) \left( 2 \varphi \Delta_{g} \varphi -n | \nabla_{g} \varphi|^{2} \right). 
 \end{eqnarray*}
 Moreover, if $F$ is any smooth function on $(\Omega, \, \bar{g})$ we have
 \begin{eqnarray*}
 | \nabla_{\bar{g}} F |^2= \varphi^2 \sum_{k=1}^{n} \varepsilon_{k} (F_{,k})^2.
 \end{eqnarray*}
 Therefore, from \eqref{1.4} and the above equation we get
\begin{eqnarray*}
 N ^ 2 (n-1) [2 \varphi \Delta_{g} \varphi-n | \nabla_{g} \varphi|^{2}) = 2 | \nabla_{\bar{g}} \psi|^2.
 \end{eqnarray*}

In coordinates, we can infer that
	\begin{equation*}
	\sum_{k=1}^{n}\varepsilon_k\left[2(n-1) N\varphi_{,kk} -n(n-1)\frac{N}{\varphi}(\varphi_{,k})^2 - 2\frac{\varphi}{N} (\psi_{,k})^2\right]=0.
	\end{equation*}
Assuming that $\varphi(\xi)$ and $\psi(\xi)$ are functions of $\xi$ we get
\begin{equation*}
\varphi_{,k}=\varphi'\alpha_{i},\hspace{0,5cm}\varphi_{,ii}=\varphi''\alpha_{i}^{2} \hspace{0,5cm}  and \hspace{0,5cm} \psi_{,i}=\psi'\alpha_{i}.
\end{equation*}

Consequently,
\begin{equation*}
\varepsilon_{i_0}\left[2(n-1)N \varphi''-n(n-1)\frac{N}{\varphi} (\varphi')^2 - 2\frac{\varphi}{N} (\psi')^2\right]=0,
\end{equation*}
where $\varepsilon_{i_0}=\displaystyle\sum_{i}\varepsilon_{i}\alpha_{i}^{2}.$
\end{proof}

\begin{proof}[Proof of  Theorem \ref{teo1}] Let us remember the Ricci and scalar curvature for a conformal metric in the form $\bar{g}=\frac{g}{\varphi^{2}}$ (cf. \cite{khunel1988}):

\begin{align*}
Ric_{\bar{g}}=\frac{1}{\varphi^{2}} \{(n-2)\varphi \, \nabla^2_{g}\varphi+[\varphi\Delta_{g}\varphi-(n-1)|\nabla_{g}\varphi|^{2}]g \}
\quad\mbox{and}\quad
R_{\bar{g}}=(n-1)\left(2\varphi\Delta_{g}\varphi-n|\nabla_{g}\varphi|^{2}\right).
\end{align*}

By equation \eqref{003}, 
\begin{align*}
N\operatorname{Ric}_{\bar{g}} =\nabla_{\bar{g}}^2 N-2\frac{\nabla_{\bar{g}}\psi\otimes \nabla_{\bar{g}}\psi}{N}+\frac{2}{(n-1)N}|\nabla_{\bar{g}}\psi|^2\bar{g}
\end{align*}
is equivalent to
\begin{eqnarray}\label{A}
&&(n-2)N\varphi (\nabla^2_{g} \varphi)_{ij} +N[\varphi\Delta_{g}\varphi-(n-1)|\nabla_{g}\varphi|^{2}]\delta_{ij}\varepsilon_{i}\nonumber\\
&&=\varphi^2 (\nabla^2_{\bar{g}} N)_{ij}+\frac{2}{(n-1)N}|\nabla_{\bar{g}}\psi|^2\varepsilon_{i}\delta_{ij} -\frac{2\varphi^2}{N}\nabla_{\bar{g}}\psi\otimes \nabla_{\bar{g}}\psi.
\end{eqnarray}

The Hessian tensor for $\bar{g}$ is given by
\begin{align*}
(\nabla^2_{\bar{g}} N)_{ij}=N_{,ij}-\sum_{k=1}^{n}\bar{\Gamma}_{ij}^{k}N_{,k}
\end{align*}
where $\bar{\Gamma}_{ij}^{k}$ are the Christoffel symbols of the metric $\bar{g}$. For a distinct $i,j,k$, we have
\begin{align}\label{cristofel}
\bar{\Gamma}_{ij}^{k}=0,\hspace{0.5cm}\bar{\Gamma}_{ij}^{i}=-\frac{\varphi_{,j}}{\varphi},\hspace{0.5cm}\bar{\Gamma}_{ii}^{k}=\varepsilon_{i}\varepsilon_{k}\frac{\varphi_{,k}}{\varphi},\hspace{0.5cm}\bar{\Gamma}_{ii}^{i}=-\frac{\varphi_{,i}}{\varphi}.
\end{align}
Thus,

\begin{align}\label{10}
(\nabla^2_{\bar{g}} N)_{ij}=
N_{,ij}+\frac{\varphi_{,j}N_{,i}}{\varphi}+\frac{\varphi_{,i}N_{,j}}{\varphi}, \hspace{0,5cm} \mathrm{if} \hspace{0,5cm} i\neq j,
\end{align}
and
\begin{align}\label{11}
(\nabla^2_{\bar{g}} N)_{ii}=
N_{,ii}+ 2\frac{\varphi_{,i}N_{,i}}{\varphi}- \dfrac{\varepsilon_{i}}{\varphi} \sum_{k=1}^{n}\varepsilon_{k} \varphi_{,k}N_{,k}, \hspace{0,5cm}
\end{align}

Moreover, we note that
\begin{equation}\label{14}
|\nabla_{g}\varphi|^{2}=\sum_{k=1}^{n}\varepsilon_{k}\left(\varphi_{,k}\right)^{2}, \hspace{0.5cm} \Delta_{g}\varphi=\sum_{k=1}^{n}\varepsilon_{k}\varphi_{,kk},
\end{equation}
and
\begin{equation}\label{014}
(\nabla^2_{g} \varphi)_{ij}=\varphi_{,ij}, \hspace{0.5cm} (\nabla_{\bar{g}}\psi\otimes \nabla_{\bar{g}}\psi)_{ij}=\psi_{,i}\psi_{,j}.
\end{equation}

If $i\neq j$ in \eqref{A}, we obtain 
\begin{equation}\label{15}
(n-2)N (\nabla^2_{g} \varphi)_{ij}-\varphi (\nabla^2_{\bar{g}} N)_{ij}+2\frac{\varphi}{N}\psi_{,i}\psi_{,j}=0.
\end{equation}
Substituting the expressions \eqref{10} and \eqref{014} into \eqref{15} in the above PDE, we get
\begin{equation*}
(n-2)N\varphi_{,ij}-\varphi N_{,ij}-\varphi_{,j}N_{,i}-\varphi_{,i}N_{,j}+2\frac{\varphi}{N}\psi_{,i}\psi_{,j}=0.
\end{equation*}

Furthermore, if $i=j$ in \eqref{A}, from \eqref{014} we have

\begin{eqnarray}\label{13}
&&(n-2)N\varphi (\nabla^2_{g} \varphi)_{ii}+N[\varphi\Delta_{g}\varphi-(n-1)|\nabla_{g}\varphi|^{2}]\varepsilon_{i}\nonumber\\
&&=\varphi^2 (\nabla^2_{\bar{g}} N)_{ii}+\frac{2}{(n-1)N}|\nabla_{\bar{g}} \psi|^2 \varepsilon_i-\frac{2\varphi^2}{N}(\psi_{,i})^{2}.
\end{eqnarray}
Hence, combining \eqref{11}, \eqref{14} and \eqref{014} with \eqref{13} we obtain
\begin{equation*}
\varphi\left[(n-2)N\varphi_{,ii}-\varphi N_{,ii}-2\varphi_{,i}N_{,i}+2\frac{\varphi}{N}\left(\psi_{,i}\right)^2 \right]-\varepsilon_{i}(n-1)N\sum_{k=1}^{n}\varepsilon_{k}\left(\varphi_{,k}\right)^2+
\end{equation*}
\begin{equation*}
\varepsilon_{i}\sum_{k=1}^{n}\varepsilon_{k}\left[\varphi\varphi_{,kk}N+\varphi\varphi_{,k}N_{,k}-\frac{2\varphi^2}{(n-1)N}\left(\psi_{,k}\right)^2\right]=0.
\end{equation*}
 
Besides that, by Definition \ref{def1} we have $div_{\bar{g}}\left(\frac{\nabla_{\bar{g}} \psi}{N}\right)=0$ (Equation \ref{002}),
 which is equivalent to
 \begin{equation}\label{16}
 N\Delta_{\bar{g}}\psi-\bar{g}\left(\nabla_{\bar{g}}\psi,\nabla_{\bar{g}}N\right)=0.
 \end{equation}
 By definition, we get
\begin{equation*}
\bar{g}\left(\nabla_{\bar{g}}\psi,\nabla_{\bar{g}}N\right)=\bar{g}\left(\sum_{i,j}\bar{g}^{ij}\psi_{,i}\partial_j,\,\sum_{k,l}\bar{g}^{kl}N_{,k}\partial_l\right).
\end{equation*}
Since
\begin{equation*}
\frac{g}{\varphi^{2}}\left(\sum_{i,j}\varphi^{2}\varepsilon_i\delta_{ij}\psi_{,i}\partial_j,\sum_{k,l}\varphi^{2}\varepsilon_k\delta_{kl}N_{,k}\partial_l\right)=\varphi^2\sum_{i,k}\varepsilon_i  \delta_{ik}\psi_{,i}N_{,k}
\end{equation*}
we can conclude that
\begin{equation}\label{18}
\bar{g}\left(\nabla_{\bar{g}}\psi,\nabla_{\bar{g}}N\right)=\varphi^2\sum_{i}\varepsilon_i\psi_{,i}N_{,i}.
\end{equation}

The Laplacian with respect to $\bar{g}$ for any smooth function $F:\Omega\rightarrow\mathbb{R}$ is given by
\begin{eqnarray*}
\Delta_{\bar{g}}F = \sum_{i=1}^{n} \varphi^2 \varepsilon_{i} F_{,ii} - (n-2) \varphi \varphi_{,i} F_{,i}.
\end{eqnarray*}
Then, from \eqref{16} and \eqref{18} we have 
\begin{equation*}
\sum_{k=1}^{n}\varepsilon_k\{N\varphi \psi_{,kk}-(n-2)N\varphi_{,k}\psi_{,k}-\varphi\psi_{,k}N_{,k}\}=0.
\end{equation*}

Finally, by expression \eqref{001}, we get
\begin{equation*}
N\Delta_{\bar{g}} N =\widehat C^2|\nabla_{\bar{g}}\psi|^2,
\end{equation*}
consequently,
\begin{equation*}
\sum_{k=1}^{n}\varepsilon_k\{\varphi NN_{,kk} -(n-2)N \varphi_{,k}N_{,k}-\widehat C^2\varphi \left(\psi_{,k}\right)^2\}=0.
\end{equation*}
Thus, we conclude the proof of this theorem.
\end{proof}

\begin{proof}[Proof of  Theorem \ref{teo2}]
	Now, we assume that $\psi(\xi)$, $\varphi(\xi)$ and $N(\xi)$ are functions of $\xi$. The derivatives for any smooth function $F:=F(\xi)$ are given by
	\begin{equation*}
	F_{,i}=F'\alpha_{i},\hspace{0,2cm}F_{,ii}=F''\alpha_{i}^{2},\hspace{0,2cm} F_{,ij}=F''\alpha_{i}\alpha_{j}.
	\end{equation*}

	Then, by equation \eqref{004} in Theorem \ref{teo1} we obtain
	\begin{equation*}
	(n-2)N\varphi''\alpha_{i}\alpha_{j}-\varphi N''\alpha_{i}\alpha_{j}-\varphi'\alpha_{i}N'\alpha_{j}-\varphi'\alpha_{j}N'\alpha_{i}+\frac{2\varphi}{N}\psi'\alpha_i\psi'\alpha_j=0, \hspace{0,2cm} i\neq j,
	\end{equation*}
	which is equivalent to 
	\begin{equation*}
	\left[(n-2)N\varphi''-\varphi N''-2\varphi'N'+\frac{2\varphi}{N}\left(\psi'\right)^2\right]\alpha_{i}\alpha_{j}=0. \hspace{0,2cm}
	\end{equation*}
	So, for $\alpha_{i}\alpha_{j}\neq 0$, we have
	\begin{equation*}
	(n-2)N\varphi''-\varphi N''-2\varphi'N'+\frac{2\varphi}{N}\left(\psi'\right)^2=0.
	\end{equation*}
	
	Similarly, considering equation \eqref{005} from Theorem \ref{teo1}, we get the following:
	\begin{equation*}
	\varphi\left[(n-2)N\varphi''-\varphi N''-2\varphi'N'+\frac{2\varphi}{N}\left(\psi'\right)^2\right] \alpha_{i}^{2}+
	\end{equation*}
	\begin{equation*}
	\varepsilon_{i}\sum_{k=1}^{n}\varepsilon_{k}\alpha_k^2\left\{N\varphi\varphi''-(n-1)N\left(\varphi'\right)^2+\varphi\varphi'N'-\frac{2\varphi^2}{(n-1)N}\left(\psi'\right)^2\right\}=0.
	\end{equation*}
	It should be noted that $(n-2)N\varphi''-\varphi N''-2\varphi'N'+\frac{2\varphi}{N}\left(\psi'\right)^2=0$. Therefore,
	\begin{equation*}
	\varepsilon_{i_0}\left\{N\varphi\varphi''-(n-1)N\left(\varphi'\right)^2+\varphi\varphi'N'-\frac{2\varphi^2}{(n-1)N}\left(\psi'\right)^2\right\}=0.
	\end{equation*}
	
	Again, by expression \eqref{19} of Theorem \ref{teo1} we can infer that
	\begin{equation*}
	\varepsilon_{i_0}\{\varphi N\psi''-(n-2)\varphi'N\psi'-\varphi N'\psi'\}=0.
	\end{equation*}

	If, for all $i\neq j$, we have $\alpha_{i}\alpha_{j}=0$ then, without loss of generality, we can consider $\xi=x_{i_{0}}$. Thus, equation \eqref{004} is trivially satisfied for $i\neq j$.  Considering \eqref{005}, for all $i\neq i_0$ we obtain the following equation: 
	\begin{equation*}
	\varepsilon_{i_0}\left[\varphi\varphi_{,i_0i_0}N-(n-1)N\left(\varphi_{,i_0}\right)^2+\varphi\varphi_{,i_0}N_{,i_0}-\frac{2\varphi^2}{(n-1)N}\left(\psi_{,i_0}\right)^2\right]=0.
	\end{equation*}
	Therefore,
	\begin{equation*}
	\varepsilon_{i_0}\left\{N\varphi\varphi''-(n-1)N\left(\varphi'\right)^2+\varphi\varphi'N'-\frac{2\varphi^2}{(n-1)N}\left(\psi'\right)^2\right\}=0,
	\end{equation*}
	and thus, equation \eqref{4} is satisfied. 
	
	Considering $i=i_{0}$ in \eqref{005}, we obtain
	\begin{eqnarray*}
	&&\varphi\left[(n-2)N\varphi''-\varphi N''-2\varphi'N'+\frac{2\varphi}{N}\left(\psi'\right)^2\right] \nonumber\\
	&+&\left[N\varphi\varphi''-(n-1)N\left(\varphi'\right)^2+\varphi\varphi'N'-\frac{2\varphi^2}{(n-1)N}\left(\psi'\right)^2\right]=0.
	\end{eqnarray*}
	Combining the above equation with \eqref{4}, we obtain 
	\begin{equation*}
	(n-2)N\varphi''-\varphi N''-2\varphi'N'+\frac{2\varphi}{N}\left(\psi'\right)^2=0.
	\end{equation*}
	
Similarly, considering equation \eqref{19}, we get
	\begin{equation*}
	\varepsilon_{i_0}\{N\varphi \psi_{,i_0i_0}-(n-2)N\varphi_{,\varepsilon_{i_0}}\psi_{,i_0}-\varphi\psi_{,i_0}N_{,i_0}\}=0,
	\end{equation*}
which gives \eqref{20}.
	Moreover, by Equation \eqref{018} we have
	\eqref{21}.

Finally, it is easy to check that for $\sum_{i=1}^n\varepsilon_i\alpha_i^2=0$ (lightlike) the only possibility left is \eqref{3}.
\end{proof}

\begin{proof}[Proof of Theorem \ref{theorem 6}]
It follows immediately from Theorem \ref{teo2}.
\end{proof}

\begin{proof}[Proof of Theorem \ref{theorem 7}] 
 First, combining equations \eqref{4} and \eqref{23} we have 
 \begin{equation*}
     -\varphi \varphi' N'+ N \varphi \varphi'' - N (\varphi')^2= 0 \quad\Longrightarrow\quad -\frac{N'}{N} + \frac{\varphi''}{\varphi'} - \frac{\varphi'}{\varphi} =0.
 \end{equation*}
 Considering that $N$ and $\varphi$ are functions of $\xi$, integrating the above expression, we obtain $$N= k \frac{\varphi'}{\varphi},$$ with $k\in \mathbb{R}\backslash\{0\}$.

Now, from \eqref{20} we obtain
\begin{equation*}
    \frac{\psi''}{\psi'}= (n-2)\frac{\varphi'}{\varphi} + \frac{N'}{N}.
\end{equation*}
Using the fact that $N= k \frac{\varphi'}{\varphi}$, the above equation becomes  
\begin{equation*}
     \frac{\psi''}{\psi'}= (n-3)\frac{\varphi'}{\varphi} + \frac{\varphi''}{\varphi'}.
\end{equation*}
So $\psi' = k_1 \varphi' \varphi^{n-3}$ where $k_1 \in \mathbb{R}\backslash\{0\}$.
Therefore, 
\begin{equation*}
    \psi= k_1 \frac{\varphi^{n-2}}{n-2} + k_2,
   \end{equation*}
  where $k_2\in\mathbb{R}$.

Considering the equation  \eqref{3}, we have 
\begin{equation*}
    \frac{n-2}{n-1} \varphi \varphi'' N - \frac{\varphi^2 N''}{n-1} - \frac{2}{n-1} \varphi \varphi' N' + \frac{2}{(n-1) \, N} \varphi^2  (\psi')^2 = 0.
\end{equation*}
Then, combining \eqref{4} with the expression obtained previously, we get
\begin{equation}\label{*}
    \left( \frac{2n-3}{n-1} \right) \varphi''N + \left( \frac{n-3}{n-1} \right) \varphi' N' - (n-1)\frac{N}{\varphi} (\varphi')^2 - \frac{\varphi N''}{n-1 }=0.
\end{equation}

On the other hand, by equation \eqref{21}, we have 
\begin{equation} \label{21*}
        -\varphi N''+(n-2)\varphi'N'+ \left( \frac{n-2}{n-1} \right) \, \frac{2\varphi}{N}\left(\psi'\right)^2=0.
    \end{equation}
    Now, multiplying the expression \eqref{23} by $\frac{n-2}{n-1}$, we obtain 
\begin{equation}\label{23*}
    2(n-2) N \varphi'' - n(n-2) \frac{N}{\varphi} (\varphi')^2 - \left( \frac{n-2}{n-1} \right) \, \frac{2\varphi}{N}\left(\psi'\right)^2=0.
\end{equation}

Therefore, from \eqref{21*} and \eqref{23*} it follows that
\begin{equation}\label{**}
    -\varphi N''+(n-2)\varphi'N'+ 2(n-2) N \varphi'' - n(n-2) \frac{N}{\varphi} (\varphi')^2 =0.
\end{equation}
Combining \eqref{*} and \eqref{**} we can see that

\begin{eqnarray}\label{***}
\left(\frac{2n-3}{n-1} +2(n-2)\right) N\varphi'' + \left((n-2)+\frac{n-3}{n-1}\right)  \varphi'N'\nonumber\\ 
- [n-1+ n(n-2)] \frac{N}{\varphi} (\varphi')^2 - \left(1+\frac{1}{n-1} \right)\varphi N''= 0.
\end{eqnarray}
 Again, from the fact that $N= k \frac{\varphi'}{\varphi}$, we can infer that     
\begin{equation*}
    N' = k \left(  \frac{\varphi''}{\varphi} - \left( \frac{\varphi'}{\varphi
    } \right) ^2 \right)  \hspace{0,5cm} \mathrm{and} \hspace{0,5cm} N''= k \left(  \frac{\varphi'''}{\varphi} - \frac{3\varphi' \varphi''}{\varphi ^2
    }  + 2\left( \frac{\varphi'}{\varphi
    } \right) ^3 \right).
\end{equation*}
Substituting the expressions of $N$, $N'$ and $N''$ in \eqref{***}, we get
    \begin{equation*}
        \frac{-n}{n-1} \varphi''' + 3n \frac{\varphi' \varphi''}{\varphi} - n^2 \frac{(\varphi')^3}{\varphi ^2}=0
    \end{equation*}
        Therefore, 
    
      \begin{equation}\label{****}
         \frac{\varphi'''}{\varphi}  - 3(n-1) \frac{\varphi' \varphi''}{\varphi^2}+ n(n-1) \left( \frac{\varphi'}{\varphi} \right)^3 = 0.
    \end{equation}
    This finishes the proof for this theorem. The next theorem, focuses on solving the above ODE.
      \end{proof}

    \begin{proof}[Proof of  Theorem \ref{cor1}]
    
    Now, consider the following change of variable $f = \varphi'/ \varphi$ in the last theorem. We can see that this choice is related to the function $N$. This trick allows us to integrate explicitly the conformal function $\varphi$. So, 
    \begin{equation*}
     \frac{\varphi''}{\varphi} = f' + f^2  \hspace{0,5cm} \mathrm{and} \hspace{0,5cm}   \frac{\varphi'''}{\varphi}= f''+3ff'+f^3.
\end{equation*}
Now replacing the identities above in the expression \eqref{****} we get the following ODE in $f$.
\begin{equation}\label{edo_f}
        f''-3(n-2)ff'+(n-2)^2f^3=0.
    \end{equation}

Considering $v(f)=f'$, we get 
\begin{equation*}
       \frac{d^2f(\xi)}{d \xi^2}= \frac{dv(f)}{df} \frac{df}{d\xi} = v(f)v'(f).
\end{equation*}
Therefore, rewriting \eqref{edo_f} and dividing by $v(f)$, it gives us
\begin{equation}\label{chini_equation}
    v'= 3(n-2)f- (n-2)^2f^{3}v^{-1}.
\end{equation}
Note that the above ODE is from the Chini Equation family:  $v'= F(f)v^m+G(f)v+H(f)$ (see Equation 1.55, page 303 in \cite{kamke}). This type of ODE generalizes the Riccati and Abel equations. It is known that if the Chini invariant 
\begin{eqnarray*}
\alpha = F(f)^{-m-1}H(f)^{-2m+1}((F(f)H'(f) - F'(f)H(f) + mF(f)G(f)H(f))^m m^{(-m)}
\end{eqnarray*}
is independent of $f$, the substitution 
\begin{eqnarray*}
v(f)= \pm u(f)\left( \frac{H(f)}{F(f)} \right)^{1/m}
\end{eqnarray*}
leads to an ODE of separable variables.

From \eqref{chini_equation}, where $m = -1$,  we get $\alpha = \frac{-9}{2}$. So, $\alpha$ is independent of $f$, then 
\begin{equation}\label{v(f)}
    v(f) = \mp (n-2)f^2u(f),
\end{equation} 
where $u(f):=\frac{u(f)}{3}.$
Moreover,
$ v'= \mp(n-2)f \left( 2u + fu'\right).$ Hence,  rewriting \eqref{chini_equation} with  the above identities we obtain
\begin{equation*}
    \frac{udu}{(u\pm1)(u\pm1/2)}=\frac{du}{u \pm 1}-\frac{du}{2u\pm1}= \frac{- df}{f}.
\end{equation*}
Thus,
   \begin{eqnarray*}
   \ln\left|\dfrac{u\pm1}{\sqrt{2u\pm1}}\right|=\ln|f^{-1}|+c_1;\quad c_1\in\mathbb{R}.
   \end{eqnarray*}
Therefore, 
\begin{equation*}
u^2 \pm 2u(1\mp c_{1}^{2}f^{-2}) + (1\mp c_{1}^{2}f^{-2}) =0
\end{equation*}
where $c_1\neq0$ is an arbitrary constant.
Then, 
\begin{eqnarray*}
u(f) = (c_{1}^{2}f^{-2}\mp 1) \pm c_1 f^{-1}\sqrt{c_{1}^{2}f^{-2}\mp 1}.
\end{eqnarray*}
From now on, it is very important to carefully follow the signs of the equations.

Substituting $u(f)$ in \eqref{v(f)} we have
\begin{equation*}
     \frac{ df}{(c_{1}^{2}\mp f^{2}) \pm c_1 \sqrt{c_{1}^{2}\mp f^{2}}}= \mp(n-2)d\xi.
\end{equation*}
Then, by an integration we obtain
\begin{eqnarray*}
 \frac{ \pm (c_1 \mp \sqrt{c_{1}^{2}\mp f^{2}})}{c_1 f}=\mp(n-2)\xi+c_2,
\end{eqnarray*}
where $c_2\in\mathbb{R}$. From the above equation we get
\begin{eqnarray*}
\frac{\varphi'}{\varphi}=f= \frac{2c_{1}^{2}[c_2 \mp (n-2)\xi]}{1\pm c_{1}^{2}[c_2\mp(n-2)\xi]^2}.
\end{eqnarray*}

A simple integration yields to 
\begin{eqnarray}\label{phir}
\varphi(\xi)=c_3[1\pm c_{1}^{2}[c_2\mp(n-2)\xi]^2]^{-1/(n-2)},
\end{eqnarray}
where $c_3$ is a non null constant.

Therefore, from \eqref{Npsidependephi} we can infer that
\begin{equation*}
   \psi(\xi)=\frac{k_1}{n-2}c_3^{n-2}[1\pm c_{1}^{2}[c_2\mp(n-2)\xi]^2]^{-1} + k_2
\end{equation*}

and

\begin{equation*}
     N\left(\xi\right)= \frac{2kc_{1}^{2}[c_2 \mp (n-2)\xi]}{1\pm c_{1}^{2}[c_2\mp(n-2)\xi]^2},
\end{equation*}
where $k_2,\,c_2,\, k\neq0,\,k_1\neq0,\,c_3\neq0,\, c_1\neq0 \in\mathbb{R}$.

 Now, to get the converse statement for this theorem we need to prove that $N$, $\psi$ and $\varphi$ are, indeed, solutions for Theorem \ref{teo2} and Lemma \ref{lemma 1}. A straightforward computation proves that $\varphi$ satisfies \eqref{edophi}. On the other hand, to get \eqref{23}, or the equations in Theorem \ref{teo2}, the following identity must be satisfied
 
 \begin{eqnarray}\label{ufa}
 c_1^2= \frac{\mp c_3^{2(n-2)}k_1^{2}}{2(n-1)(n-2)k^{2}}.
 \end{eqnarray}

In fact, from the equation in \eqref{23}, knowing that $N= k \frac{\varphi'}{\varphi}$  and $\psi' = k_1 \varphi' \varphi^{n-3}$, we obtain 
 
 \begin{equation}\label{23_varphi}
     2(n-1) \frac{\varphi' \varphi''}{\varphi}-n(n-1)\frac{(\varphi ')^3}{\varphi^2} - 2\frac{k_1^2}{k^2} \varphi' \varphi^{2(n-2)} =0.
 \end{equation}
 Considering \eqref{phir}, we have
 
 \begin{equation*}
     \varphi'(\xi) = \frac{2c_1^2 c_3(c_2 \mp (n-2)\xi)}{[1\pm c_{1}^{2}[c_2\mp(n-2)\xi]^2]^{(n-1)/(n-2)}}.
 \end{equation*}
 and 
  \begin{equation*}
     \varphi''(\xi) = \frac{2c_1^2 c_3 \{ nc_1^2(c_2 \mp (n-2)\xi)^2 \mp (n-2)\}}{[1\pm c_{1}^{2}[c_2\mp(n-2)\xi]^2]^{(2n-3)/(n-2)}}.
 \end{equation*}
Therefore, substituting the expressions above in \eqref{23_varphi}, from a simple computation we obtain the following identity:
\begin{eqnarray*}
\dfrac{8c_1^4 c_3 (c_2 \mp (n-2)\xi)} {[1\pm c_{1}^{2}[c_2\mp(n-2)\xi]^2]^{(3n-5)/(n-2)}} \left\{ \mp (n-1)(n-2)-\frac{k_1^2c_3^{2(n-2)}}{2 k^2 c_1^2} \right\} =0.
\end{eqnarray*}

Similarly, we can verify that $N$, $\psi$ and $\varphi$ satisfy the equations of Theorem \ref{teo2} only if \eqref{ufa} is satisfied. Hence, the only possibility for $\varphi,\,\psi$ and $N$ is the one provided in the statement of this theorem, i.e., \eqref{varphi}, \eqref{psi} and \eqref{N}, respectively. 
\end{proof}

\section{Appendix}

In this section, we will provide the formula for the sectional curvature for a conformal metric. Consider the metric $g_{ij}= \frac{\delta_{ij}\varepsilon_{i}}{\varphi^2}$ in $\Omega\subseteq\mathbb{R}^n$ a open subset, where $\varphi$ is a positive smooth function. We can write $g^{ij}= \varphi^2\delta_{ij}\varepsilon_{i}$ to indicate the inverse of the metric $g_{ij}$. In these conditions, we have:
\begin{eqnarray*}
 \frac{\partial g_{ik}}{\partial x_j} = \frac{-2 \delta_{ik} \varepsilon_i}{\varphi^3} \varphi_{,j}.  
\end{eqnarray*}

For this metric $g$, the curvature coefficients are
 \begin{equation*}
     R_{ijij}= \sum_{l} R_{iji}^l \, g_{lj} = R_{iji}^j \, \frac{\varepsilon_j}{\varphi^2} = \frac{\varepsilon_j}{\varphi^2} \left( \sum_l \Gamma_{ii}^{l} \Gamma_{jl}^{j} - \sum_l \Gamma_{ji}^{l} \Gamma_{il}^{j} + \frac{\partial}{\partial x_j}\Gamma_{ii}^{j} - \frac{\partial}{\partial x_i}\Gamma_{ji}^{j} \right).
 \end{equation*}
Then, from \eqref{cristofel} we can calculate the derivative of $\Gamma_{ij}^{k}$. That is, 
 \begin{equation*}
     \frac{\partial}{\partial x_j}\Gamma_{ii}^{j}= \varepsilon_i\varepsilon_j \left( \frac{\varphi_{,jj}}{\varphi} - \frac{\varphi_{,j} ^2 }{\varphi^2}\right) \hspace{0,5cm} \mathrm{and} \hspace{0,5cm} \frac{\partial}{\partial x_i}\Gamma_{ji}^{j}= \left( \frac{\varphi_{,i}}{\varphi} \right)^2  - \frac{\varphi_{,ii}}{\varphi}.
 \end{equation*}
Combining the above identities with \eqref{cristofel} we get
\begin{equation*}
    \varphi^2\varepsilon_j  R_{ijij}= - \sum_l \varepsilon_i\varepsilon_l \left( \frac{\varphi_{,l}}{\varphi} \right)^2 + \varepsilon_i\varepsilon_j \frac{\varphi_{,jj}}{\varphi} + \frac{\varphi_{,ii}}{\varphi}.
    \end{equation*}
 Now, if the four indices are distinct
 
 \begin{equation*}
     R_{ijkl}=  \sum_{s} R_{ijk}^s \, g_{sl} = R_{ijk}^l \, g_{ll} = 0.
 \end{equation*}
 
 Finally, consider the case in which we have three distinct indices:
 \begin{equation*}
      R_{ijk}^i= - \frac{\varphi_{,kj}}{\varphi}, \hspace{0,5cm}
      R_{ijk}^j= \frac{\varphi_{,ki}}{\varphi} \hspace{0,5cm} \mathrm{and} \hspace{0,5cm}      R_{ijk}^k= 0.
 \end{equation*}
Hence, the sectional curvature generated by $\partial_{x_i}$, $\partial_{x_j}$ is
 \begin{equation}\label{Kij}
     K_{ij} = \varphi^2\left(-\sum_l \varepsilon_i\varepsilon_l \left( \frac{\varphi_{,l}}{\varphi} \right)^2 + \varepsilon_i\varepsilon_j \frac{\varphi_{,jj}}{\varphi} + \frac{\varphi_{,ii}}{\varphi}\right)\varepsilon_i.
 \end{equation}

Let $\xi=\sum_{i=1}^{n}\alpha_ix_i$, consider $\varphi(\xi)$ a function of $\xi$. Since
\begin{equation*}
\varphi_{,i}=\varphi'\alpha_i, \hspace{0.5cm}\varphi_{,ii}=\varphi''\alpha_i^2, \hspace{0.5cm}\varphi_{,ij}=\varphi''\alpha_i\alpha_j,
\end{equation*}
from \eqref{Kij} we get 
 \begin{equation*}
     K_{ij} = -\left(\varphi'\right)^2\varepsilon_{i_0}  + \varphi\varphi''\varepsilon_j \alpha_j^2 + \varphi\varphi''\varepsilon_i\alpha_i^2,
 \end{equation*}
 where $\displaystyle\sum_{l=1}^{n} \varepsilon_{l}\alpha_{l}^{2}=\varepsilon_{i_0}\in\{-1,\,0,\,1\}$, which depends on the direction of the tangent vector field $\alpha=\displaystyle\sum_{l=1}^{n}\alpha_{l}\partial_{x_l}$.

\

\


\

\end{document}